\documentclass[11pt, a4paper]{article}
\usepackage{amsmath,amssymb,amsthm}
\usepackage{tikz}

\usepackage[auth-sc, affil-it]{authblk}

\usepackage[utf8]{inputenc}
\usepackage[english]{babel}
\usepackage[T1]{fontenc}
\usepackage{lmodern}
\usepackage{microtype}
\usepackage[
    linktocpage=true,
    colorlinks=true,
    linkcolor=purple,
    citecolor=black,
    urlcolor=purple,
    bookmarksnumbered=true,
    pdftitle={Undecidable properties of self-affine sets and multi-tape automata},
    pdfauthor={Jarkko Kari and Timo Jolivet},
]{hyperref}

\usepackage[totalwidth=160mm, totalheight=247mm]{geometry}

\usepackage{enumitem}
\setlist{noitemsep}
\setlist[enumerate]{label=(\arabic*)}

\let\oldbibliography\thebibliography
\renewcommand{\thebibliography}[1]{%
  \oldbibliography{}%
  \small%
  \setlength{\itemsep}{0pt}%
}

\makeatletter
\def\@cite#1#2{[{{\bfseries#1}\if@tempswa , #2\fi}]} 
\renewcommand{\@biblabel}[1]{[{\bfseries{#1}}]~} 
\makeatother

\usepackage[format=plain,justification=justified]{caption}
\theoremstyle{plain}
\newtheorem{theo}{Theorem}[section]
\newtheorem{prop}[theo]{Proposition}
\newtheorem{lemm}[theo]{Lemma}

\theoremstyle{definition}
\newtheorem{defi}[theo]{Definition}
\newtheorem{exam}[theo]{Example}
\newtheorem{rema}[theo]{Remark}

\renewcommand{\leq}{\leqslant}
\renewcommand{\geq}{\geqslant}

\newcommand{\bbN}{{\mathbb N}}
\newcommand{\bbZ}{{\mathbb Z}}

\newcommand{\bbR}{{\mathbb R}}

\newcommand{\mcA}{\mathcal A}
\newcommand{\mcB}{\mathcal B}

\newcommand{\mcD}{\mathcal D}

\newcommand{\mcG}{\mathcal G}

\newcommand{\mcM}{\mathcal M}

\newcommand{\mcQ}{\mathcal Q}
\newcommand{\mcR}{\mathcal R}

\DeclareTextFontCommand{\tdef}{\itshape\bfseries} 

\title{\textbf{
Undecidable properties of self-affine sets \\
and multi-tape automata}}
\author[1,2]{Timo Jolivet}
\author[1]{Jarkko Kari}
\date{}
\affil[1]{%
Department of Mathematics, University of Turku, Finland
}
\affil[2]{%
LIAFA, Universit\'e Paris Diderot, France
}

\begin{document}
\maketitle

\begin{abstract}
We study the decidability of the topological properties
of some objects coming from fractal geometry.
We prove that having empty interior is undecidable
for the sets defined by two-dimensional graph-directed iterated function systems.
These results are obtained by studying a particular class of self-affine sets associated with multi-tape automata.
We first establish the undecidability of some language-theoretical properties of such automata,
which then translate into undecidability results about their associated self-affine sets.
\end{abstract}

\section{Introduction}
\label{sect:intro_undfrac}

A classical way to define fractals
is to use an \tdef{iterated function system (IFS)},
specified by a finite collection of maps $f_1, \ldots, f_n : \bbR^d \rightarrow \bbR^d$
which are all \tdef{contracting}:
there exists $0 \leq c < 1$ such that $\|f_i(x) - f_i(y)\| \leq c \|x-y\|$ for all $x, y \in \bbR^d$.
The associated fractal set, called the \tdef{attractor}
of the IFS, is the unique nonempty compact set $X$ such that
$X = \bigcup_{i=1}^n f_i(X)$.
Such a set $X$ always exists and is unique thanks to a famous result of Hutchinson~\cite{Hut81},
based on an application of the Banach fixed-point theorem; see also~\cite{Fal03} or~\cite{Bar93}.
For example, the classical Cantor set can be defined as the unique compact set $X \subseteq \bbR$
satisfying the set equation $X = \frac{1}{3}X \cup (\frac{1}{3}X + \frac{2}{3})$.
Two other examples are given in Figure~\ref{fig:exifs}.

\begin{figure}
\centering
\includegraphics[height=50mm]{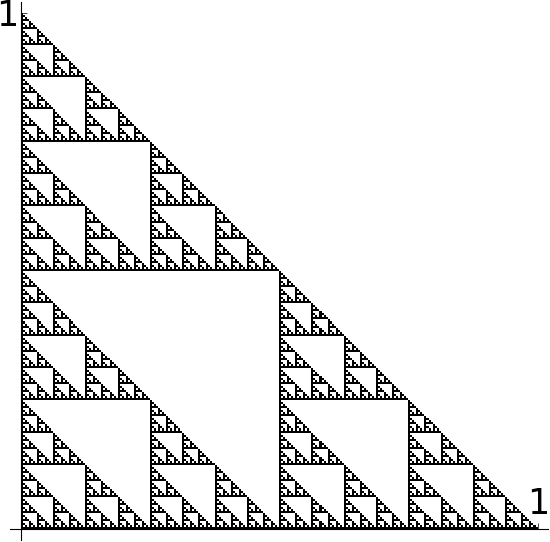}
\hfil
\includegraphics[height=50mm]{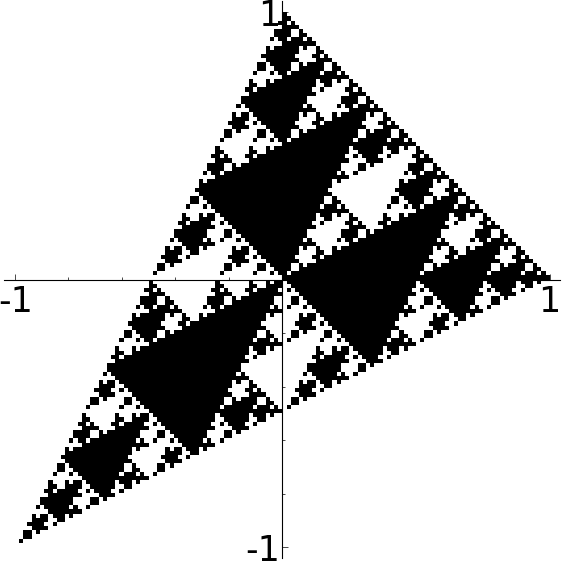}
\caption[Two self-affine sets]{
Two self-affine sets defined by
$\smash{X = \bigcup_{v \in \mcD} M^{-1}(X + v)}$,
where $M = (\begin{smallmatrix}2 & 0 \\ 0 & 2\end{smallmatrix})$
and $\mcD = \{(\begin{smallmatrix}0 \\ 0\end{smallmatrix}),
              (\begin{smallmatrix}1 \\ 0\end{smallmatrix}),
              (\begin{smallmatrix}0 \\ 1\end{smallmatrix})\}$ (left),
and $\mcD = \{(\begin{smallmatrix}0 \\ 0\end{smallmatrix}),
              (\begin{smallmatrix}1 \\ 0\end{smallmatrix}),
              (\begin{smallmatrix}0 \\ 1\end{smallmatrix}),
              (\begin{smallmatrix}-1 \\ -1\end{smallmatrix})\}$ (right).
The set on the left is the Sierpi\'nski triangle and has empty interior.
The set on the right is an example of a self-affine tile with nonempty interior
(see~\cite{LW96}).
}
\label{fig:exifs}
\end{figure}

A question of interest is to determine when the fractal set $X$ has nonempty interior.
This question arises in several areas,
including tiling theory, dynamical systems, number theory and Fourier analysis
(see~\cite{W99,LLR13}
and references therein).
A well studied case is when the contracting maps are affine mappings
of the form $f_i(x) = M^{-1}(x + v_i)$ where $v_i \in \bbZ^d$
and $M$ is an integer expanding matrix which is common to all the $f_i$,
like in the examples of Figure~\ref{fig:exifs}.
In this case, having nonempty interior is equivalent with having nonzero Lebesgue measure,
and there are efficient algorithms to decide this~\cite{GY06,BK11cool}.

Much less is known in the more general case of \tdef{self-affine attractors},
where the maps $f_i$ are only restricted to be affine
(of the form $f_i = M_ix~+~v_i$ where the~$M_i$ are real matrices and $v_i \in \bbR^d$).
No algorithm is known to decide nonempty interior in this case,
and specific results such as computation of Hausdorff dimension
are known only for some very specific families of self-affine sets~\cite{Bed84,McM84,Fra13}.

\paragraph{Our results}
We are interested in the following question:
to what extent can we decide if a self-affine set has nonemtpy interior?

We will answer this question by an undecidability result
for a natural generalization of iterated function systems,
which consist of a \emph{finite system} of equations instead of just one,
hence defining a finite number of attractors.
This is formalized in the following definition:
a $d$-dimensional \tdef{graph-directed iterated function system (GIFS)}
is a directed graph
in which each edge $e$ is labelled by a contracting mapping $f_e : \bbR^d \rightarrow \bbR^d$.
The \tdef{attractors} of the GIFS
are the unique nonempty compact sets $\{X_q\}_{q \in \mcQ}$ such that
\[
X_q = \bigcup_{r \in \mcQ} \bigcup_{e \in E_{q,r}} f_e(X_r),
\]
where $\mcQ$ is the set of vertices of the directed graph defining the GIFS,
and $E_{q,r}$ denote the set of edges from vertex $q$ to vertex $r$.
Again, such a collection of compact sets $\{X_q\}_{q \in \mcQ}$
exists and is unique~\cite{Fal97}.
Fractals defined by GIFS are widely used to define various self-similar tilings of the plane,
the study of which have applications in physics, dynamics and number theory.
Note that the case of single-vertex graphs corresponds to
classical iterated function systems.

Our main result, Theorem~\ref{theo:und_emptyint},
states that it is undecidable if the attractors
of a $2$-dimensional, $3$-state affine GIFS have empty interior.
We follow an approach initiated by Dube~\cite{Dub93}
by associating self-affine sets with finite multi-tape automata.
Then we relate some properties of the automaton with topological properties of its associated attractor,
and we obtain the undecidability of the latter by proving the undecidability of the former.
The original motivation of~\cite{Dub93} is to prove that it is undecidable
if a if the attractor of a rational $2$-dimensional affine IFS
intersects the diagonal $\{(x,x) : x \in [0,1]\}$.

In Section~\ref{sect:multi-tape} we define multi-tape automata
and we consider a variant of the Post correspondence problem in Section~\ref{subsec:pcp},
which we use in Section~\ref{subsec:undaut} to prove
undecidability results about multi-tape automata.
We then relate some language-theoretical properties of an automaton
with some topological properties of its attractor in Section~\ref{sect:mtagigs}.
The main results are stated in~Section~\ref{sect:mainresults_undfrac}.

\paragraph{Acknowledgements.}
Research supported by the Academy of Finland Grant 131558
and by project Fractals and Numeration ANR-12-IS01-0002.

\section{Multi-tape automata}
\label{sect:multi-tape}

\subsection{Definitions}

A \tdef{$d$-tape automaton} $\mcM$ on alphabet $\mcA = A_1 \times \cdots \times A_d$
is defined by a finite set of \tdef{states} $\mcQ$,
and a finite set of \tdef{transitions} $\mcR \subseteq \mcQ \times \mcQ \times (A_1^+ \times \cdots \times A_d^+)$.
A $d$-tape automaton on state $\mcQ$ is conveniently represented by a directed graph
with vertex set $\mcQ$ and an edge $(q,r)$ labelled by $w_1 | \cdots | w_d$
for every transition $(q,r,(w_1, \ldots, w_d))$.
This is illustrated in Example~\ref{exam:mta}.

A \tdef{configuration} is an infinite sequence $c \in \mcA^\bbN = (A_1 \times \cdots \times A_d)^\bbN$.
For $k \in \{1, \ldots, d\}$, the \tdef{$k$th tape} of $c$ refers to the infinite sequence $((c_n)_k)_{n \in \bbN}$,
which is an infinite concatenation of words in $A_k^\star$.
For convenience, configurations will be denoted by writing the tape components separated by the symbol ``|''.
For example, $00\cdots \,|\, 11\cdots \,|\, 00\cdots$
denotes the $3$-tape configuration $(0,1,0),(0,1,0),\ldots \in (\{0,1\} \times \{0,1\} \times \{0,1\})^\bbN$.

Let $q$ be a state of $\mcM$.
A configuration $c \in A^\bbN$ is \tdef{$q$-accepted} by $\mcM$
if there exists an infinite sequence of transitions
$((q_n,r_n,(w_{n,1},\ldots,w_{n,d})))_{n \geq 1}$
such that $q_1 = q$, $r_n = q_{n+1}$ for all $n \geq 1$,
and for every $k \in \{1, \ldots, d\}$,
the infinite word $w_{1,k}w_{2,k}\ldots$
is equal to the $k$th tape of $c$
(that is, $w_{1,k}w_{2,k}\ldots = (c_1)_k(c_2)_k \ldots$).
Such an infinite sequence of transitions will sometimes be referred to as
a \tdef{run of $\mcM$ starting at~$q$}.
Note that we also forbid $\varepsilon$-transitions
as the words $w_1, \ldots, w_d$ used in transitions are nonempty,
so that each infinite run provides an infinite word on every tape.

\begin{exam}
\label{exam:mta}
Consider the following $2$-tape, $2$-state automaton on alphabet $\mcA = \{0,1\} \times \{0,1,2\}$,
with state set $\mcQ = \{X,Y\}$ and transitions given by the following.
\begin{center}
\begin{tikzpicture}[x={13mm}, y={13mm}]
\draw (0,2) node[minimum size=6mm, inner sep=0pt,       draw,circle,thick, very thick] (X) {$X$};
\draw (2,2) node[minimum size=6mm, inner sep=0pt,       draw,circle,thick, very thick] (Y) {$Y$};
\draw[->, very thick] (X) .. controls +(150:10mm) and +(210:10mm) .. node [left] {\small$0 | 22$} (X);
\draw[->, very thick] (Y) .. controls +(10:10mm) and +(70:10mm) .. node [right] {\small$1|001$} (Y);
\draw[->, very thick] (Y) .. controls +(-70:10mm) and +(-10:10mm) .. node [right] {\small$20|1$} (Y);
\draw[<-, very thick] (Y) .. controls +(150:10mm) and +(30:10mm) .. node [above] {\small$10|11$} (X);
\draw[<-, very thick] (X) .. controls +(-30:10mm) and +(-150:10mm) .. node [below] {\small$110|2$} (Y);
\end{tikzpicture}
\end{center}
It is easy to check that the configuration $00\cdots | 22\cdots$
is not $Y$-accepted
but is $X$-accepted by $\mcM$ (by repeatedly using the transition $(X,X,(0,22))$).
However, giving a precise description of the set of configurations
which are accepted by $\mcM$ seems difficult.
\end{exam}

\begin{rema}
Multi-tape automata are very powerful computational devices
because of the fact that the words $w_1, \ldots, w_d$ in a transition are allowed to have different lengths.
This is the fundamental feature that will allow us to establish several undecidability results about multi-tape automata
later in this section.
On the other hand, if the words $w_1, \ldots, w_d$ all have the same length in every transition,
then it is easy to see that this model is not more powerful than a classical finite automaton on a product alphabet.
\end{rema}

\subsection{Post correspondence problems}
\label{subsec:pcp}

The undecidability results of this article are all derived from the undecidability of the following decision problems.
The \tdef{Post correspondence problem (PCP)} is:
    given $n$ pairs of nonempty words $(u_1,v_1), \ldots, (u_n,v_n)$,
    decide if there exist $m\geq1$ and a word $i_1 \cdots i_m$ such that
    $u_{i_1} \cdots u_{i_m} = v_{i_1} \cdots v_{i_m}$.
This is a well-known undecidable problem~\cite{Pos46}.

We will need a slight variant of PCP,
the \tdef{prefix Post correspondence problem (prefix-PCP)}:
    given $n$ pairs of nonempty words $(u_1,v_1), \ldots, (u_n,v_n)$,
    decide if there exist $m,m' \geq 1$ and two words $i_1 \cdots i_m$ and $i_1 \cdots i_{m'}$ such that
    $\smash{u_{i_1} \cdots u_{i_m} = v_{i_1} \cdots v_{i_{m'}}}$
    and one of the two words $i_1 \cdots i_m$ and $i_1 \cdots i_{m'}$ is a prefix of the other.

A positive PCP always yields a positive prefix-PCP instance (by taking $m = m'$),
but the converse is not always true.
For example, the instance $(u_1,v_1) = (a,abb)$, $(u_2,v_2) = (bb,aa)$
admits the prefix-PCP solution given by $u_1u_2u_1u_1 = v_1v_2 = abbaa$,
that is, $m = 4$, $m' = 2$ and the two words $i_1i_2i_3i_4 = 1211$ and $i_1i_2 = 12$.
However, this instance cannot admit any PCP solution because no pair of words ends by the same symbol.

\begin{lemm}
\label{lemm:und_prefixPCP}
Prefix-PCP is undecidable.
\end{lemm}

\begin{proof}
We reduce PCP to prefix-PCP.
Let $(u_1,v_1), \ldots, (u_n,v_n)$ be an instance of PCP on alphabet $\mcA$.
Let $\mcB = \mcA \cup \{\texttt\#, \texttt*\}$ be a new alphabet,
where $\texttt\#$ and $\texttt*$ are two new symbols not contained in $\mcA$.
We construct a prefix-PCP instance
$(A_1,B_1)$, \ldots, $(A_n,B_n)$, $(U_1,V_1)$, \ldots, $(U_n,V_n)$, $(Y, Z)$
on the new alphabet $\mcB$, defined by
\begin{align*}
A_i &= \texttt\# x_1 \texttt* x_2 \texttt* \cdots \texttt* x_k
 & U_i &= \texttt* x_1 \texttt* x_2 \texttt* \cdots \texttt* x_k
 & Y &= \texttt*\texttt\# \\
B_i &= \texttt\# y_1 \texttt* y_2 \texttt* \cdots \texttt* y_{\ell} \texttt*
 & V_i &= \phantom{\texttt*}\, y_1 \texttt* y_2 \texttt* \cdots \texttt* y_{\ell} \texttt*
 & Z &= \phantom{\texttt*}\texttt\#
\end{align*}
for all $i \in \{1, \ldots, n\}$, where
$u_i = x_1 \cdots x_n$ and $v_i = y_1 \cdots y_\ell$
and the $x_j, y_j$ are in $\mcA$.
We now prove that the PCP instance has a solution
if and only if the prefix-PCP instance has a solution.
Suppose that there exists a solution $i_1 \cdots i_m$ to the PCP instance,
that is $u_{i_1} \cdots u_{i_m} = v_{i_1} \cdots v_{i_m}$.
Then clearly the prefix-PCP also has a solution, given by
$A_{i_1} U_{i_2} \cdots U_{i_m} Y = B_{i_1} V_{i_2} \cdots V_{i_m} Z$.

Conversely, suppose that the prefix-PCP instance has a solution.
By construction, because of $\texttt\#$ and $\texttt*$,
there must exist a prefix-PCP solution of the form
$A_{i_1} U_{i_2} \cdots U_{i_m} Y = B_{i_1} V_{i_2} \cdots V_{i_{m'}} Z$,
where $i_1 \cdots i_m$ is a prefix of $i_1 \cdots i_{m'}$ or vice-versa.
But the pairs $(U_i,V_i)$ do not contain any $\texttt\#$,
so the pair $(Y,Z)$ is used exactly once,
both after $m$th pair and the $m'$th pair,
so $m = m'$ and the PCP instance has a solution.
\end{proof}

\subsection{Undecidable properties of multi-tape automata}
\label{subsec:undaut}

Let $\mcM$ be a $d$-tape automaton
on alphabet $\mcA$,
and let $q$ be a state of $\mcM$.
State $q$ is \tdef{universal}
if every sequence in $\mcA^\bbN$ is $q$-accepted by $\mcM$.
A finite sequence $x \in \mcA^\star$
is a \tdef{universal prefix}
for state $q$ if for every infinite sequence $y \in \mcA^\bbN$,
the infinite sequence $xy$ is $q$-accepted by $\mcM$.

\begin{exam}
\label{exam:punotu}
Let $\mcM$ be a $1$-tape, $1$-state automaton on alphabet $\{0,1\}$
with three transitions labelled by $1$, $10$ and $00$.
The single state of $\mcM$ is not universal because every sequence starting with $01$ is rejected,
but the word $1$ is a universal prefix:
any sequence starting with $1$ is accepted,
because any finite segment $10^n1$ is accepted by transitions
$1$, $00 \times k$, $1$ if $n=2k$
or $10$, $00 \times k$, $1$ if $n=2k+1$,
and any infinite tail of $0$'s or $1$'s is obviously accepted.
Hence there exist some multi-tape automata without universal states
but that admit universal prefixes.
The self-affine set associated with this automaton is discussed
in Example~\ref{exam:punotu_attractor}.
\end{exam}

\begin{theo}
\label{theo:unduniv}
It is undecidable whether
a given state of a given $d$-tape automaton is universal.
This problem remains undecidable if we restrict to $2$-tape automata with $3$ states.
\end{theo}

\begin{proof}
We reduce prefix-PCP,
which is undecidable thanks to Lemma~\ref{lemm:und_prefixPCP}.
Let $(u_1,v_1)$, \ldots, $(u_n,v_n)$ be an instance of prefix-PCP where the $u_i$, $v_i$ are words over $\mcB$.
We define a $2$-tape automaton $\mcM$ on $3$ states (denoted by $X,U,V$).
The alphabet of $\mcM$ is $A_1 \times A_2$,
with $A_1 = \{1, \ldots, n\}$ and $A_2 = \mcB \cup \{\texttt\#\}$,
where $n$ is the size of the prefix-PCP instance,
$\mcB$ is the alphabet of words $u_i,v_i$
and $\texttt\#$ is a new symbol not in $\mcB$.
The transitions of $\mcM$ are
\begin{enumerate}
\item\label{item:univred1}
$X \stackrel{i | u_i}{\longrightarrow} U$ and $U \stackrel{i | u_i}{\longrightarrow} U$ for every $i \in A_1$;
\item\label{item:univred2}
$X \stackrel{i | v_i}{\longrightarrow} V$ and $V \stackrel{i | v_i}{\longrightarrow} V$ for every $i \in A_1$;
\item\label{item:univred3}
$U \stackrel{i | u}{\longrightarrow} X$ for every $i \in A_1$ and $u \in A_2^+$ such that
    \begin{enumerate}[label=(\roman*)]
    \item\label{item:univred31u} $|u| \leq |u_i|$,
    \item\label{item:univred32u} $u$ is not a prefix of $u_i$,
    \item\label{item:univred33u} $u$ does not begin with $\texttt\#$;
    \end{enumerate}
\item\label{item:univred4}
$X \stackrel{i | u}{\longrightarrow} X$ for every $i \in A_1$ and $u \in A_2^+$
such that~\ref{item:univred31u} and~\ref{item:univred32u} above hold;
\item\label{item:univred5}
$V \stackrel{i | v}{\longrightarrow} X$ for every $i \in A_1$ and $v \in A_2^+$ such that
    \begin{enumerate}[label=(\roman*)]
    \item\label{item:univred31v} $|v| \leq |v_i|$,
    \item\label{item:univred32v} $v$ is not a prefix of $v_i$,
    \item\label{item:univred33v} $v$ does not begin with $\texttt\#$;
    \end{enumerate}
\item\label{item:univred6}
$X \stackrel{i | v}{\longrightarrow} X$ for every $i \in A_1$ and $v \in A_2^+$
such that~\ref{item:univred31v} and~\ref{item:univred32v} above hold.
\end{enumerate}
We now prove that the prefix-PCP instance $(u_1,v_1), \ldots, (u_n,v_n)$
has a solution if and only if state $X$ is \textbf{not} universal in $\mcM$.

($\Rightarrow$)
Suppose that the prefix-PCP instance admits a solution:
there exist $m,m' \geq 1$ and two words $i_1 \cdots i_m$ and $i_1 \cdots i_{m'}$ such that
$u_{i_1} \cdots u_{i_m} = v_{i_1} \cdots v_{i_{m'}}$
and one of the two words $i_1 \cdots i_m$ and $i_1 \cdots i_{m'}$ is a prefix of the other.
Without loss of generality we can assume that $m \geq m'$ and $i_1 \cdots i_{m'}$ is a prefix of $i_1 \cdots i_m$.
We prove that $\mcM$ cannot accept any infinite sequence in $(A_1 \times A_2)^\bbN$ beginning with
\[
i_1 \cdots i_m \ | \ u_{i_1} \cdots u_{i_m}\texttt\#
\]
when starting from state $X$, so $\mcM$ is not universal.
Indeed, let us describe the evolution of $\mcM$ when reading such a sequence.
\begin{itemize}
\item
We start from $X$, so $\mcM$ necessarily uses a transition defined in~\ref{item:univred1} and~\ref{item:univred2}
and moves to state $U$ or $V$ after having read $i_1 | u_{i_1}$ or $i_1 | v_{i_1}$, respectively.
(The other transitions~\ref{item:univred4} and~\ref{item:univred6} cannot be used
because of the conditions~\ref{item:univred31u} and~\ref{item:univred32u}.)
Note that both $u_{i_1}$ and $v_{i_1}$ are prefixes of the content of the second tape.
\item
Now if $\mcM$ is in state $U$, the remaining input starts with some $i$ on the first tape
and starts with $u_i$ on the second tape.
So $\mcM$ must use transition~\ref{item:univred1}: stay in state $U$ and read $i | u_i$.
(Transition~\ref{item:univred3} cannot be used
because of the conditions~\ref{item:univred31u} and~\ref{item:univred32u}.)
The same holds if $\mcM$ is in state $V$.
\end{itemize}
It follows that when $\mcM$ reads $i_1, \ldots, i_{m'}$ on the first tape,
then it is either in state $U$ and has read $u_{i_1} \cdots u_{i_{m'}}$ on the second tape,
or it is in state $V$ and has read
$v_{i_1} \cdots v_{i_{m'}} = u_{i_1} \cdots u_{i_m}$ on the second tape.
In the second case, the next symbol on the second tape is $\texttt\#$,
so $\mcM$ is ``blocked'' on this input
(there is no suitable transition for this sequence because of~\ref{item:univred33u}).
In the first case, the computation must continue in the same way as before,
so eventually $\mcM$ is still in state $U$ and has read
$i_1 \cdots i_m | u_{i_1} \cdots u_{i_m}$,
and again, $\mcM$ is blocked because the next symbol on the second tape is $\texttt\#$.

($\Leftarrow$)
Suppose that no solution exists for the prefix-PCP instance.
The following strategy shows that a move by the automaton can always be made,
whatever its tape contents is.
If $\mcM$ is in state $U$ or $V$, make any available move.
In state $X$, if no loop in $X$ is possible,
then in the current configuration $(i_1i_2 \cdots | w)$,
both $u_{i_1}$ and $v_{i_1}$ must be prefixes of $w$,
otherwise~\ref{item:univred4} or~\ref{item:univred6} could have been used.
Write $w = u_{i_1} w' = v_{i_1} w''$.
Then:

\begin{enumerate}[label=(\alph*)]
\item if $u_{i_1} \cdots u_{i_k} \texttt\#$ is a prefix of $w$ for some $k$, do not go to $U$ by reading $i_1 | u_{i_1}$;
\label{item:stratuniva}
\item if $v_{i_1} \cdots v_{i_k} \texttt\#$ is a prefix of $w$ for some $k$, do not go to $V$ by reading $i_1 | v_{i_1}$;
\label{item:stratunivb}
\end{enumerate}
The only possible ways to be stuck at this point are:
\begin{itemize}
\item $\mcM$ is in state $U$ or $V$ and the next symbol on the second tape is $\texttt\#$;
\item $\mcM$ is in state $X$ and~\ref{item:stratuniva},~\ref{item:stratunivb} prevent from moving to $U$ or $V$.
\end{itemize}
The second case cannot happen because it implies the existence of a prefix-PCP solution.
If we are in the first case, we can assume by symmetry that we are in state $U$.
In the last step where $\mcM$ went from $X$ to $U$,
the configuration must start with $i_1i_2 \cdots | u_{i_1}u_{i_2} \cdots u_{i_k}\texttt\#\cdots$ for some $k$,
because this is the only way to get stuck in $U$ some $k$ steps later.
However, this contradicts the choice made in~\ref{item:stratuniva} above,
because $\mcM$ should have moved to $V$ instead of state $U$.
\end{proof}

\begin{theo}
\label{theo:undunivpref}
It is undecidable whether
a given state of a given $d$-tape automaton admits a universal prefix.
This problem remains undecidable if we restrict to $2$-tape automata with $3$ states.
\end{theo}

\begin{proof}
We modify the prefix-PCP reduction made in the proof of Theorem~\ref{theo:unduniv}.
Let $(u_1,v_1), \ldots, (u_n,v_n)$ be an instance of prefix-PCP where the $u_i$, $v_i$ are words over $\mcB^\star$.
First we modify the $u_i,v_i$ by adding a new symbol $\texttt*$ not in $\mcB$
after each letter of each $u_i$ and each $v_i$
(a word $x_1x_2 \cdots x_k$ becomes $x_1\texttt*x_2\texttt* \cdots x_k\texttt*$).
This modified instance is clearly equivalent to the original one,
so we denote it again by $(u_1,v_1), \ldots, (u_n,v_n)$.

We now define a $2$-tape automaton $\mcM$ on $3$ states $X,U,V$.
We take the same alphabet $A_1 \times A_2$ as in the other reduction,
with a new symbol $\texttt\&$ for both $A_1$ and $A_2$,
and the symbol $\texttt*$ for $A_2$.
This gives
$A_1 = \{1, \ldots, n\} \cup \{\texttt\&\}$
and $A_2 = \mcB \cup \{\texttt\#, \texttt\&, \texttt*\}$,
where $n$ is the size of the prefix-PCP instance,
$\mcB$ is the alphabet of the words $u_i,v_i$
and $\texttt\#, \texttt\&, \texttt*$ are new symbol not in $\mcB$.
The transitions of $\mcM$ consist of
\begin{itemize}
\item \ref{item:univred1} and~\ref{item:univred2} like in the proof of Theorem~\ref{theo:unduniv},
without allowing any symbol $\texttt\&$ or $\texttt*$;
\item
\ref{item:univred3},~\ref{item:univred4},~\ref{item:univred5},~\ref{item:univred6}
like in the proof of Theorem~\ref{theo:unduniv},
where symbols $\texttt\&$ or $\texttt*$ are allowed, except in the first letter of $u$ or $v$;
\end{itemize}
plus the following transitions:
\begin{enumerate}
\setcounter{enumi}{6}
\item\label{item:univred7}
$X \stackrel{a | \texttt\&}{\longrightarrow} X$,
$U \stackrel{a | \texttt\&}{\longrightarrow} X$ and
$V \stackrel{a | \texttt\&}{\longrightarrow} X$
for every $a \in A_1$;
\item\label{item:univred8}
$X \stackrel{\texttt\& | a}{\longrightarrow} X$,
$U \stackrel{\texttt\& | a}{\longrightarrow} X$ and
$V \stackrel{\texttt\& | a}{\longrightarrow} X$
for every $a \in A_2 \setminus \{\texttt*\}$;
\item\label{item:univred9}
$X \stackrel{a | \texttt*b}{\longrightarrow} X$,
$U \stackrel{a | \texttt*b}{\longrightarrow} X$ and
$V \stackrel{a | \texttt*b}{\longrightarrow} X$
for every $a \in A_1$ and $b \in A_2$.
\end{enumerate}
We now prove that the prefix-PCP instance $(u_1,v_1), \ldots, (u_n,v_n)$
has a solution if and only if state $X$ does \textbf{not} have any universal prefix.

($\Rightarrow$)
Suppose that the prefix-PCP instance has a solution:
there exist $m,m' \geq 1$ and two words $i_1 \cdots i_m$ and $i_1 \cdots i_{m'}$ such that
$u_{i_1} \cdots u_{i_m} = v_{i_1} \cdots v_{i_{m'}}$
and one of the two words $i_1 \cdots i_m$ and $i_1 \cdots i_{m'}$ is a prefix of the other.
Consider the following claim.
\begin{quote}
\textbf{{Claim.}}
Let $x \in A_1^\star$ and $y \in A_2^\star$ be such that
$x \texttt\&\texttt\& \cdots | y \texttt\&\texttt\& \cdots$
is $X$-accepted by at most $k \geq 1$ different runs of $\mcM$.
Then there exist $x' \in A_1^\star$ and $y' \in A_2^\star$ such that
$x x' \texttt\&\texttt\& \cdots | y y' \texttt\&\texttt\& \cdots$
is $X$-accepted by at most $k-1$ different runs.
\end{quote}
This claim implies that $X$ does not have any universal prefix,
\emph{i.e.}, that for every finite words $x \in A_1^\star$ and $y \in A_2^\star$,
there exists a configuration starting with $x | y$ that is not $X$-accepted.
Indeed, for every such $x,y$, there can be only finitely many different accepting runs (say~$k$)
for $x \texttt\&\texttt\& \cdots | y \texttt\&\texttt\& \cdots$,
because $\mcM$ eventually loops on state $X$ with transition $\texttt\& | \texttt\&$.
So it suffices to apply the claim $k$ times to obtain a configuration starting with $x | y$
which is not $X$-accepted.

We now prove the claim, using the prefix-PCP solution.
Let $x \in A_1^\star$ and $y \in A_2^\star$ be such that
$x \texttt\&\texttt\& \cdots | y \texttt\&\texttt\& \cdots$
is $X$-accepted by $k$ different runs.
Denote by $R_1, \ldots, R_k$ the finite prefixes of the $k$ runs,
each cut when $\mcM$ reaches the $\texttt\&\texttt\&\cdots | \texttt\&\texttt\&\cdots$ part.
Let $s = i_1 \cdots i_m \in A_1^\star$
and let $t = u_1 \cdots u_{i_m}$,
which can be written in the form $t = a_1 \texttt* a_2 \texttt* \cdots \texttt* a_{|t|-1} \texttt* \in A_2^\star$,
where each $a_i$ is in $A_2 \setminus \{\texttt\#,\texttt\&,\texttt*\}$,
thanks to the modification made to the instance.

Let $\ell$ be the distance between the two tapes heads when $\mcM$ has completed the finite run $R_1$.
(Note that the first head is always behind the second one because it can only move by one cell at at time.)
Without loss of generality we can assume that $R_1$ is the run for which such an $\ell$ is minimal.
We now construct a configuration $c$ which will ``block'' any run starting with $R_1$,
without giving the other runs any possibilities for new nondeterministic branching.
Let $L,L' \geq 0$ such that $s$ (on the first tape) begins $\ell$ positions behind $t$ (on the second tape)
in the configuration
\[
c \ = \ x \texttt\&^L    s           \texttt\&\texttt\& \cdots
  \ | \ y \texttt\&^{L'} t \texttt\# \texttt\&\texttt\& \cdots,
\]
so that during any run starting with $R_1$,
$\mcM$ starts reading $s$ and $t \texttt\#$ exactly at the same time.
It follows that $R_1$ cannot be extended to an accepting run for $c$,
because $s,t$ corresponds to a prefix-PCP solution,
similarly as in the proof of Theorem~\ref{theo:unduniv}.
The same is true for any other run $R_i$ for which such an $\ell$ is the same as $R_1$.

Let us now consider another accepting run $R_i$.
By minimality of $\ell$, the distance between the two tapes heads when $\mcM$ first reaches
$\texttt\&\texttt\&\cdots | \texttt\&\texttt\&\cdots$
during run $R_i$ is strictly larger than $\ell$.
We now prove that $R_i$ can be extended in a unique way to an accepting run for $c$.
Indeed, any run of $\mcM$ starting with $R_i$ must evolve in the following way:
\begin{itemize}
\item
when $t$ starts being read the second tape, $s$ is not yet being read on the first tape,
so at this time $\mcM$ is reading $\texttt\&$ on the first tape and $a_1$ on the second tape;
\item
the only possible transition is~\ref{item:univred8}, so $\mcM$ moves one step on both tapes,
and is now reading $\texttt*$ on the second tape;
\item
the only possible transition is~\ref{item:univred9}, so $\mcM$ moves one step on the first tape and two steps on the second,
and is again reading $\texttt*$ on the second tape;
\item
this continues until the whole $t = a_1 \texttt* a_2 \texttt* \cdots \texttt* a_{|t|-1} \texttt*$
has been read on the second tape,
and $\mcM$ is deterministically looping on $\texttt\& | \texttt\&$.
\end{itemize}
From this analysis, it follows that $R_i$ can be extended in a \emph{unique} way to an accepting run for $c$.
Hence $c$ is a configuration starting with $x | y$ with at most $k-1$ accepting runs,
because every accepting run for $c$ must start with an $R_i$,
each of which can be extended in at most one way if $i \in \{2, \ldots, k\}$,
or in no way at all if $i = 1$.
Thus the claim is proved by taking
$x' = \texttt\&^L s$ and $y' = \texttt\&^{L'} t \texttt\#$.

($\Leftarrow$)
Suppose that no solution exists for the prefix-PCP instance.
The strategy described in the ``$\Leftarrow$'' direction of the proof of Theorem~\ref{theo:unduniv}
can be applied to prove that state $X$ is universal,
with the additional case that if the tape begins by $\texttt\&$ or $\texttt*$,
then the transition~\ref{item:univred7},~\ref{item:univred8} or~\ref{item:univred9} can always be used.
\end{proof}

\begin{rema}
\label{rema:commonprefix}
In the reduction made in the above proof of Theorem~\ref{theo:undunivpref},
if state $X$ has a universal prefix, then in fact $X$ is universal.
Also, in this case, it is easy to see that any finite word satisfying
\ref{item:univred31u},~\ref{item:univred32u} and~\ref{item:univred33u}
of transition~\ref{item:univred3} is a universal prefix for $U$ (and $V$),
so $X$, $U$ (and $V$) have a common universal prefix
Hence we have the following:
given a $2$-tape automaton $\mcM$ on $3$ states and two states $q,r$ of $\mcM$,
it is undecidable if $q$ and $r$ have a common universal prefix.
\end{rema}

\section{Affine GIFS associated with multi-tape automata}
\label{sect:mtagigs}

Let $\mcM$ be a $d$-tape automaton on alphabet $\mcA = A_1 \times \ldots \times A_d$.
We want to give a ``numerical interpretation'' to a finite word $u \in \mcA^\star$
or to an infinite configuration $c \in \mcA^\bbN$.
We must first specify, for each $k \in \{1, \ldots, n\}$,
a numerical interpretation of the letters of $A_k$
by choosing a bijection $\delta_k : A_k \rightarrow \{0, \ldots, |A_k|-1\}$.
We then define $\Delta_k : A_k^\star \rightarrow \bbR$ by
\[
\Delta_k(u) = \sum_{1 \leq i \leq |u|} \delta_k(u_i) |A_k|^{-i}.
\]
Equivalently, for $u = u_1 \cdots u_n \in A_k^n$,
the number $\Delta_k(u)$ is represented by $0.\delta_k(u_1) \cdots \delta_k(u_n)$ in base $|A_k|$.
Finally, let $\Delta : A_1^+ \times \ldots \times A_d^+ \rightarrow \bbR^d$
be defined by $\Delta(w_1,\ldots, w_d) = (\Delta_1(w_1), \ldots, \Delta_d(w_d))$.
The domains of $\Delta_k$ and $\Delta$ can naturally be extended to $A_k^\bbN$ and $\mcA^\bbN$, respectively.

In the examples that will follow,
if the alphabets $A_k$ are all of the form $\{0, \ldots, |A_k|-1\}$
and the maps $\delta_k : A_k \rightarrow \{0, \ldots, |A_k|-1\}$ are not specified,
we will assume for convenience that they are identity mappings.

\begin{defi}
Let $\mcM$ be a $d$-tape automaton on state set $\mcQ$ and alphabet $\mcA = A_1 \times \cdots \times A_n$.
The \tdef{GIFS associated with $\mcM$} is the GIFS defined by
the graph $G$ with vertex set $\mcQ$
and, for every transition $R = (q,r,(w_1, \ldots, w_d))$ of $\mcM$,
an edge $(q,r)$ labelled by the map $f_R : [0,1]^d \rightarrow [0,1]^d$ defined by
\[
f_R(x) \ = \
\begin{pmatrix}
|A_1|^{-|w_1|} & & 0 \\
  & \ddots & \\
0 & & |A_d|^{-|w_d|}
\end{pmatrix}x
\ + \
\Delta(w_1, \ldots, w_d).
\]
\end{defi}

\begin{exam}
Let $\mcM$ be a $2$-tape automaton on alphabet $\mcA = \{0,1\} \times \{0,1\}$,
and let $c \in \mcA^\bbN$ be configuration.
If $\mcM$ contains a transition $R = (q,r,(1011,11))$,
then applying the contracting map $f_R$ on $\Delta(c) = (0.x_1x_2\ldots, 0.y_1y_2\ldots) \in [0,1]^2$ has the following effect:
\begin{align*}
f_R(\Delta(c))
    &= \begin{pmatrix}1/16 & 0 \\ 0 & 1/4\end{pmatrix}
       \begin{pmatrix}0.x_1x_2\ldots \\ 0.y_1u_2\ldots \end{pmatrix}
       + \Delta(1011,11) \\
    &= \begin{pmatrix}0.0000x_1x_2\ldots \\ 0.00y_1u_2\ldots\phantom{00} \end{pmatrix}
       + \begin{pmatrix}0.1011 \\ 0.11\phantom{00} \end{pmatrix} 
    = \begin{pmatrix}0.1011x_1x_2\ldots \\ 0.11y_1u_2\ldots \phantom{00}\end{pmatrix}.
\end{align*}
This suggests that applying a sequence of mappings $f_{R_1} \cdots f_{R_n}(\Delta(c))$
corresponds to concatenating the words associated with the transitions $R_n$
in the numerical interpretation $\Delta(c)$ of a configuration $c$.
This is the key idea to establish a correspondence
between the GIFS of an automaton and its accepted configurations.
This is formalized in the next proposition.
\end{exam}

\begin{prop}
\label{prop:aut_address}
Let $\mcM$ be a $2$-tape automaton and let $q$ be a state of $\mcM$.
The GIFS attractor of $\mcM$ associated with $q$
is equal to the set
$\{\Delta(c) \in \bbR^d : c \in \mcA^\bbN \text{ is $q$-accepted by } \mcM\}$.
\end{prop}

\begin{proof}
Let $x \in [0,1]^d$.
It follows from a standard fact in the theory of iterated function systems~\cite[Chapter~9]{Fal03}
that $x \in X_q$ if and only if
there is an infinite run $(R_n)_{n\geq 1}$ starting at $q$
such that $x = \bigcap_{n \geq 1} f_{R_1} \cdots f_{R_n}([0,1]^d)$,
where $f_{R_n}$ is the mapping of the GIFS of $\mcM$ associated with run $R_n$.
Moreover, by definition of the GIFS of $\mcM$,
for every such run $(R_n)_{n\geq 1}$, the configuration
$c = w_{1,1}w_{2,1}\cdots \ | \ \cdots \ | \ w_{1,d}w_{2,d}\cdots$
is such that $x = \Delta(c)$,
where the $w_{n,k}$ are given by the transitions
$(q_n,r_n,(w_{n,1},\ldots,w_{n,d}))$ for all $n \geq 1$,
so the proposition is proved
because $c$ is a $q$-accepted configuration.
\end{proof}

\begin{exam}
Let $\mcM$ be the $1$-state, $2$-tape automaton on alphabet $\{0,1\}$
with transitions $0 | 0$, $0 | 1$, $1 | 0$.
The iterated function system associated with $\mcM$ consists of the maps
$x \mapsto \frac{x}{2}$,
$x \mapsto \frac{x}{2} + (\frac{1}{2},0)$,
$x \mapsto \frac{x}{2} + (0,\frac{1}{2})$
and it can easily be seen that the associated attractor
the Sierpi\'nski triangle (see Figure~\ref{fig:exifs}).
%
\end{exam}

\begin{exam}
\label{exam:punotu_attractor}
The $1$-tape, $1$-state automaton $\mcM$ on alphabet $\{0,1\}$
with three transitions $1$, $10$ and $00$
(given in Example~\ref{exam:punotu})
is an example of a non-universal automaton which admits universal prefixes.
This reflects in the attractor associated with $\mcM$ in the following way:
it is not equal to $[0,1]$ but it has nonempty interior.
This can be proved either by Proposition~\ref{prop:aut_to_topo},
or by proving directly that a configuration $x \in \{0,1\}^\bbN$
is accepted by $\mcM$ if and only if it does not start
with $0^{2k+1}1$ for some $k \geq 0$,
which implies that the attractor is equal to $\bigcup_{k \geq 0} [2^{-2k-1}, 2^{-2k}]$.
\end{exam}

\begin{rema}
\label{rema:finitetoone}
Given a $d$-tape automaton and a point $x \in [0,1]^d$,
if there exists two configurations $c,c'$ such that $x = \Delta(c) = \Delta(c')$
and such that the tapes components $c_k \in A_k^\bbN$ and $c'_k \in A_k^\bbN$ differ for some $k \in \{1, \ldots, d\}$,
then $c_k$ and $c_k'$ are both stationary, ending with $0^\omega$ or $(|A_k|-1)^\omega$.
In particular, $\Delta : \mcA^\bbN \rightarrow \bbR^d$ is finite-to-one.
\end{rema}

The next proposition establishes the desired correspondence
between word-theoretical properties of multi-tape automata
and topological properties of the associated self-affine attractors.

\begin{prop}
\label{prop:aut_to_topo}
Let $\mcM$ be a $d$-tape automaton on alphabet $\mcA$,
let $q$ be a state of $\mcM$,
and let $X_q$ be the associated GIFS attractor.
We have:
\begin{enumerate}
\item
\label{item:toprop1}
$q$ is universal if and only if $X_q = [0,1]^d$,
\item
\label{item:toprop2}
$q$ has a universal prefix if and only if $X_q$ has nonemtpy interior.
\end{enumerate}
\end{prop}

\begin{proof}
\ref{item:toprop1}
If state $q$ is universal
the expansion of every element of $[0,1]^d$ is $q$-accepted
so $X_q = [0,1]^d$ thanks to Proposition~\ref{prop:aut_address}.
Conversely, suppose that there exists an infinite sequence $c$ that is not $q$-accepted.
By a compactness argument, there must exist a prefix $w$ of $c$ such that
$wc'$ is not $q$-accepted for any infinite sequence $c'$.
Thanks to Remark~\ref{rema:finitetoone},
by choosing $c'$ with no tape components ending by $0^\omega$ or $(|A_k|-1)^\omega$,
the sequence $wc'$ is the \emph{only} sequence such that $x = \Delta(wc')$,
so $\Delta(wc') \notin X_q$ because otherwise $wc'$ would be $q$-accepted.
It follows that $X_q \neq [0,1]^d$.

\ref{item:toprop2}
For a finite word $w \in \mcA^\star$, define the \emph{cylinder} $[w]$
to be equal to the set of configurations that start with $w$.
If $q$ admits a universal prefix $w$, then $\Delta([w]) \subseteq X_q$ by Proposition~\ref{prop:aut_address},
so $X_q$ has nonempty interior.
Conversely, suppose that there exists a nonempty open set $U \subseteq X_q$,
and let $w \in \mcA^\star$ be a finite word such that $\Delta([w]) \subseteq U$.
By a reasoning similar as in the proof of~\ref{item:toprop1},
we can prove that $w$ is a universal prefix for $q$.
\end{proof}

\section{Undecidability results for self-affine sets}
\label{sect:mainresults_undfrac}

Thanks to the undecidability results obtained for multi-tape automata in Theorem~\ref{theo:unduniv}
and to the correspondence between word-theoretical and topological properties in Proposition~\ref{prop:aut_to_topo},
we obtain the following undecidability results about topological properties of self-affine attractors.

The first result below states that it is undecidable if an attractor ``takes up the whole space'',
that is, equals $[0,1]^d$.
It follows directly from Theorem~\ref{theo:unduniv}
and Proposition~\ref{prop:aut_to_topo},~\ref{item:toprop1}.

\begin{theo}
\label{theo:und_eqsquare}
The following problem is undecidable.
Instance: a $d$-dimensional affine GIFS $\mcG$ specified by maps with rational coefficients,
and a state $q$ of $\mcG$.
Question: is $X_q = [0,1]^d$?
This problem remains undecidable if we restrict to $2$-dimensional GIFS with $3$ states.
\end{theo}

The next result states the undecidability of a fundamental topological property
for self-affine sets: having empty interior.
It is a direct corollary of Theorem~\ref{theo:undunivpref}
and Proposition~\ref{prop:aut_to_topo},~\ref{item:toprop2}.

\begin{theo}
\label{theo:und_emptyint}
The following problem is undecidable.
Instance: a $d$-dimensional affine GIFS $\mcG$ specified by maps with rational coefficients,
and a state $q$ of $\mcG$.
Question: does $X_q$ have empty interior?
This problem remains undecidable if we restrict to $2$-dimensional GIFS with $3$ states.
\end{theo}

\begin{rema}
\label{rema:intpow}
All the undecidability results above have been obtained via a reduction using
affine GIFS associated with a multi-tape automaton.
Hence it follows that undecidability holds even if we restrict
to affine GIFS in which the linear part of the contractions $f_i$
are diagonal matrices whose entries are negative powers of integers.
By adding dummy duplicate symbols,
undecidability holds even if the entries are negative powers of two.
\end{rema}

\begin{rema}
We can deduce from Remark~\ref{rema:commonprefix} that the following problem is undecidable.
Instance: a $d$-dimensional affine GIFS $\mcG$ specified by maps with rational coefficients,
and two states $q, r$ of $\mcG$.
Question: does $X_q \cap X_r$ have empty interior?
Indeed, it can be shown that $q$ and $r$ have a common universal prefix if and only if $X_q \cap X_r$ has nonemtpy interior,
similarly as in Proposition~\ref{prop:aut_to_topo}.
\end{rema}

\section{Conclusion}
\label{sect:conclu_undfrac}

We conclude this article by some questions and perspectives for further work.
Is nonempty interior decidable for $1$-state GIFS?
(That is, for classical affine IFS.)
What about the $1$-dimensional case?
Using multi-tape automata may lead to an undecidability result for the $1$-state case,
but for not for the $1$-dimensional case.
Indeed, $1$-tape automata are not more powerful than classical finite automata,
for which the properties we used in this article are all decidable.
Note that for $1$-state multi-tape automata, universality is trivially decidable,
but the status of prefix-universality is not known in this case.

Also, let us note that having nonempty interior is equivalent to having nonzero Lebesgue measure
in the case of integer self-affine tiles (as mentioned in the introduction),
but not in the more general setting of self-affine (G)IFS (see for example~\cite{CJPPS06}).
How do these properties relate in the case of self-affine sets arising from multi-tape automata?

Another interesting aspect is the computability of fractal dimension (such as Hausdorff dimension).
For example, can we decide if the Hausdorff dimension of a $2$-dimensional self-affine set is equal to $2$?
And in the case of a self-affine set with nonempty interior, can we compute the Hausdorff dimension of its boundary?
Very few results are known in this direction, apart from some very specific families
such as Bedford-McMullen carpets~\cite{Bed84,McM84,Fra13}.
A possible approach towards undecidability would be to adapt the reductions of this article
in such a way that the Hausdorff dimension can be controlled in the reductions,
or to relate the entropy of the automaton language with the Hausdorff dimension of its attractor
and prove that entropy is uncomputable.

\bibliographystyle{amsalpha}
\bibliography{biblio}
\end{document}